\ifx\GenSoCGVer\undefined%
\documentclass[12pt]{article}%
\newcommand{\RegVer}[1]{#1}
\newcommand{\SoCGVer}[1]{}
\else%

\documentclass[english]{socg/socg-lipics-v2019}%
\newcommand{\RegVer}[1]{}%
\newcommand{\SoCGVer}[1]{#1}%
\fi%
\SoCGVer{\usepackage{subcaption}}%

\providecommand{\BibLatexMode}[1]{}
\providecommand{\BibTexMode}[1]{#1}

\ifx\UseBibLatex\undefined%
  \renewcommand{\BibLatexMode}[1]{}
  \renewcommand{\BibTexMode}[1]{#1}
\else
  \renewcommand{\BibLatexMode}[1]{#1}
  \renewcommand{\BibTexMode}[1]{}
\fi

\BibLatexMode{%
   \usepackage[bibencoding=ascii,style=alphabetic,backend=biber]{biblatex}%
   \UsePackage{sariel_biblatex}%
}

\newcommand{\SarielComp}[1]{}
\newcommand{\IfPrinterVer}[2]{#2}%

\RegVer{%
   \usepackage[in]{fullpage}%
}%
\usepackage{amsmath}%
\usepackage{amssymb}%
\usepackage{xcolor}%
\usepackage{mathbbol}%

\RegVer{%
   \usepackage[amsmath,thmmarks]{ntheorem}%
   \theoremseparator{.}%
   
   \usepackage{titlesec}%
   \titlelabel{\thetitle. }%
}

\usepackage{xcolor}%
\usepackage{mleftright}%
\usepackage{xspace}%
\usepackage{hyperref}%

\usepackage{graphicx}

\definecolor{blue25emph}{rgb}{0, 0, 11}

\DeclareFontFamily{U}{BOONDOX-calo}{\skewchar\font=45 }
\DeclareFontShape{U}{BOONDOX-calo}{m}{n}{
  <-> s*[1.05] BOONDOX-r-calo}{}
\DeclareFontShape{U}{BOONDOX-calo}{b}{n}{
  <-> s*[1.05] BOONDOX-b-calo}{}
\DeclareMathAlphabet{\mathcalb}{U}{BOONDOX-calo}{m}{n}
\SetMathAlphabet{\mathcalb}{bold}{U}{BOONDOX-calo}{b}{n}
\DeclareMathAlphabet{\mathbcalb}{U}{BOONDOX-calo}{b}{n}

\newcommand{\hrefb}[3][black]{\href{#2}{\color{#1}{#3}}}%

\IfPrinterVer{%
   \usepackage{hyperref}%
}{%
   \usepackage{hyperref}%
   \hypersetup{%
      breaklinks,%
      ocgcolorlinks, colorlinks=true,%
      urlcolor=[rgb]{0.25,0.0,0.0},%
      linkcolor=[rgb]{0.5,0.0,0.0},%
      citecolor=[rgb]{0,0.2,0.445},%
      filecolor=[rgb]{0,0,0.4},
      anchorcolor=[rgb]={0.0,0.1,0.2}%
   }
}

\SoCGVer{%
   \theoremstyle{remark}%
   \newtheorem{problem}[theorem]{Problem}%
   \newtheorem{defn}[theorem]{Definition}%
   \newtheorem{assumption}[theorem]{Assumption}%
   \newtheorem{observation}[theorem]{Observation}%
   \newtheorem{fact}[theorem]{Fact}%

   \newtheorem*{remark:unnumbered}{Remark}%
}%

\RegVer{%
\theoremseparator{.}%

\theoremstyle{plain}%
\newtheorem{theorem}{Theorem}[section]

\newtheorem{lemma}[theorem]{Lemma}

\newtheorem{observation}[theorem]{Observation}

\theoremstyle{plain}%
\theoremheaderfont{\sf} \theorembodyfont{\upshape}%
\newtheorem*{remark:unnumbered}[theorem]{Remark}%
\newtheorem{remark}[theorem]{Remark}%
\newtheorem{definition}[theorem]{Definition}

\newcommand{\myqedsymbol}{\rule{2mm}{2mm}}

\theoremheaderfont{\em}%
\theorembodyfont{\upshape}%
\theoremstyle{nonumberplain}%
\theoremseparator{}%
\theoremsymbol{\myqedsymbol}%
\newtheorem{proof}{Proof:}%

}

\newcommand{\atgen}{\symbol{'100}}
\newcommand{\SarielThanks}[1]{\thanks{Department of Computer Science;
      University of Illinois; 201 N. Goodwin Avenue; Urbana, IL,
      61801, USA; {\tt sariel\atgen{}illinois.edu}; {\tt
         \url{http://sarielhp.org/}.} #1}}

\newcommand{\HLink}[2]{\hyperref[#2]{#1~\ref*{#2}}}
\newcommand{\HLinkSuffix}[3]{\hyperref[#2]{#1\ref*{#2}{#3}}}

\newcommand{\thmlab}[1]{{\label{theo:#1}}}
\newcommand{\thmref}[1]{\HLink{Theorem}{theo:#1}}

\newcommand{\itemlab}[1]{\label{item:#1}}
\newcommand{\itemref}[1]{\HLinkSuffix{}{item:#1}{}}

\newcommand{\apndlab}[1]{\label{apnd:#1}}
\newcommand{\apndref}[1]{\HLink{Appendix}{apnd:#1}}

\newcommand{\remlab}[1]{\label{rem:#1}}
\newcommand{\remref}[1]{\HLink{Remark}{rem:#1}}%

\newcommand{\seclab}[1]{\label{sec:#1}}
\newcommand{\secref}[1]{\HLink{Section}{sec:#1}}

\providecommand{\deflab}[1]{\label{def:#1}}
\newcommand{\defref}[1]{\HLink{Definition}{def:#1}}

\newcommand{\lemlab}[1]{\label{lemma:#1}}
\newcommand{\lemref}[1]{\HLink{Lemma}{lemma:#1}}%

\providecommand{\eqlab}[1]{}%
\renewcommand{\eqlab}[1]{\label{equation:#1}}
\newcommand{\Eqref}[1]{\HLinkSuffix{Eq.~(}{equation:#1}{)}}

\providecommand{\remove}[1]{}%
\newcommand{\Set}[2]{\left\{ #1 \;\middle\vert\; #2 \right\}}
\newcommand{\pth}[2][\!]{\mleft({#2}\mright)}%

\newcommand{\ceil}[1]{\left\lceil {#1} \right\rceil}

\newcommand{\brc}[1]{\left\{ {#1} \right\}}
\newcommand{\cardin}[1]{\left| {#1} \right|}%

\renewcommand{\th}{th\xspace}

\renewcommand{\Re}{\mathbb{R}}%
\usepackage[inline]{enumitem}

\newlist{compactenumN}{enumerate}{5}%
\setlist[compactenumN]{topsep=0pt,itemsep=-1ex,partopsep=1ex,parsep=1ex,%
   label=\arabic*.}%

\setlist[compactenumN,2]{label={\arabic{compactenumNi}.\arabic*.}}

\newlist{compactenumA}{enumerate}{5}%
\setlist[compactenumA]{topsep=0pt,itemsep=-1ex,partopsep=1ex,parsep=1ex,%
   label=(\Alph*)}%
\setlist[compactenumA,2]{label={(\Alph{compactenumAi}.\roman*)}}

\newlist{compactenuma}{enumerate}{5}%
\setlist[compactenuma]{topsep=0pt,itemsep=-1ex,partopsep=1ex,parsep=1ex,%
   label=(\alph*)}%

\newlist{compactenumI}{enumerate}{5}%
\setlist[compactenumI]{topsep=0pt,itemsep=-1ex,partopsep=1ex,parsep=1ex,%
   label=(\Roman*)}%

\newlist{compactenumi}{enumerate}{5}%
\setlist[compactenumi]{topsep=0pt,itemsep=-1ex,partopsep=1ex,parsep=1ex,%
   label=(\roman*)}%

\newlist{compactitem}{itemize}{5}%
\setlist[compactitem]{topsep=0pt,itemsep=-1ex,partopsep=1ex,parsep=1ex,%
   label=\bullet}%

\providecommand{\Mh}[1]{#1}%

\newcommand{\PS}{\Mh{P}}%
\newcommand{\PSA}{\Mh{Q}}%

\newcommand{\ca}{\Mh{c}}%
\newcommand{\cb}{\Mh{f}}%
\newcommand{\cc}{\Mh{g}}%

\newcommand{\pa}{\Mh{p}}%
\newcommand{\pb}{\Mh{q}}%
\newcommand{\pc}{\Mh{t}}%

\newcommand{\clusterY}[2]{\mathsf{cl}\pth{#1,#2}}%

\newcommand{\dSY}[2]{\Mh{\mathsf{d}}\pth{#1,#2}}%
\newcommand{\dY}[2]{\left\|#1 - #2 \right\|}%

\newcommand{\mrgY}[2]{\Mh{\nabla}\pth{#1, #2}}
\newcommand{\mrgC}{\Mh{\nabla}}%

\newcommand{\ts}{\hspace{0.6pt}} %
\newcommand{\profitC}{\Mh{\rho}}%
\newcommand{\popt}{\Mh{\rho}^{}_{\Opt}}%
\newcommand{\profitX}[1]{\profitC\pth{ #1}}%
\newcommand{\profitY}[2]{\profitC\pth{ #1, #2}}%

\newcommand{\optY}[2]{\mathrm{opt}^{}_{#1}\pth{ #2}}%

\newcommand{\CS}{\Mh{C}}%
\newcommand{\CSA}{\Mh{D}}%
\newcommand{\Copt}{\Mh{C^*}}%

\newcommand{\copt}{\Mh{c^*}}%

\newcommand{\SaveContent}[2]{%
   \expandafter\newcommand{#1}{#2}%
}

\newcommand{\RestatementOf}[2]{
   \noindent%
   \textbf{Restatement of #1.}
   {\em #2{}}%
}

\numberwithin{figure}{section}%
\numberwithin{table}{section}%
\numberwithin{equation}{section}%

\newcommand{\eps}{{\varepsilon}}%
\newcommand{\ArrX}[1]{\mathcal{A}\pth{#1}}%

\newcommand{\ball}{\Mh{\mathbb{b}}}%

\newcommand{\SepSet}{\Mh{Z}}%

\newcommand{\VorX}[1]{\Mh{\mathcal{V}}\pth{#1}}%
\newcommand{\VorCell}[2]{\Mh{\mathcal{C}}_{#2} \pth{#1}}%

\newcommand{\distChar}{\mathsf{d}}%
\newcommand{\distSetY}[2]{\distChar\pth{#1, #2}}
\newcommand{\BSet}{\Mh{\partial}}%
\newcommand{\Div}{\Mh{\mathcal{D}}}%
\newcommand{\inSet}{\Mh{I}}%
\newcommand{\GS}{\Mh{G}}%
\newcommand{\batch}{\Mh{B}}%
\newcommand{\CC}[1]{\textcolor{blue}{\texttt{#1}}}

\newcommand{\bSize}{\Mh{\alpha}}%
\newcommand{\constA}{\Mh{\mathsf{c}_1}}%
\newcommand{\constB}{\Mh{\mathsf{c}_2}}%
\newcommand{\constD}{\Mh{\mathsf{c}_4}}%
\newcommand{\xs}{\Mh{\xi}}%
\newcommand{\balance}{\Mh{\Delta}}%

\newcommand{\lSize}{\Mh{\mathcalb{l}}}%
\newcommand{\oSize}{\Mh{\mathcalb{o}}}%

\newcommand{\constS}{\Mh{\gamma}}

\newcommand{\Lopt}{\Mh{\mathsf{L}}}%
\newcommand{\LMX}[1]{\Mh{\overline{\mathsf{L}}_{#1}}}%
\newcommand{\LM}{\Mh{\overline{\mathsf{L}}}}%

\newcommand{\Opt}{\Mh{\mathsf{O}}}%
\newcommand{\OMX}[1]{\Mh{\overline{\mathsf{O}}_{#1}}}%

\newcommand{\OM}{\Mh{\overline{\mathsf{O}}}}%

\newcommand{\US}{\Mh{\mathcal{U}}}%
\newcommand{\etal}{\textit{et~al.}\xspace}

\newcommand{\Term}[1]{\textsf{#1}}

\newcommand{\PTAS}{\Term{PTAS}\xspace}%

\newcommand{\sFunc}{\Mh{\mathsf{\varphi}}}%

\newcommand{\ub}{\Mh{\mathsf{u}}}%

\newcommand{\nnY}[2]{\mathsf{n{}n}\pth{#1,#2}}%
\newcommand{\Family}{\Mh{\mathcal{F}}}%
\newcommand*{\biga}[1]{\vcenter{\hbox{\scalebox{1.5}{\ensuremath#1}}}}
\newcommand{\si}[1]{#1}

\allowdisplaybreaks

\title{Submodular Clustering in Low Dimensions}

\date{\today}%
\author{%
   Arturs Backurs\thanks{Toyota Technological Institute at Chicago;
      {\tt backurs\atgen{}ttic.e{d}u}.}%
   \and%
   Sariel Har-Peled\SarielThanks{Work on this paper was partially
      supported by a NSF AF awards %
      CCF-1421231, %
      and %
      CCF-1907400.  %
   }}%

\begin{document}

\maketitle

\begin{abstract}
    We study a clustering problem where the goal is to maximize the
    coverage of the input points by $k$ chosen centers. Specifically,
    given a set of $n$ points $\PS \subseteq \Re^d$, the goal is to
    pick $k$ centers $\CS \subseteq \Re^d$ that maximize the service
    \begin{math}
        \sum_{\pa \in \PS}\sFunc\bigl( \dSY{\pa}{\CS} \bigr)
    \end{math}
    to the points $\PS$, where $\dSY{\pa}{\CS}$ is the distance of
    $\pa$ to its nearest center in $\CS$, and $\sFunc$ is a
    non-increasing service function $\sFunc : \Re^+ \to \Re^+$. This
    includes problems of placing $k$ base stations as to maximize the
    total bandwidth to the clients -- indeed, the closer the client
    is to its nearest base station, the more data it can send/receive,
    and the target is to place $k$ base stations so that the total
    bandwidth is maximized. We provide an $n^{\eps^{-O(d)}}$ time
    algorithm for this problem that achieves a $(1-\eps)$-approximation.
	Notably, the runtime does not depend on the parameter $k$ and it works for an arbitrary non-increasing service function $\sFunc : \Re^+ \to \Re^+$.

    \remove{We show that our algorithms extends to other metrics. For
       instance, if the underlying metric is a planar graph equipped
       with the shortest path distance, we show an $n^{\eps^{-O(1)}}$
       time algorithm.}
\end{abstract}

\section{Introduction}

Clustering is a fundamental problem, used almost everywhere in
computing. It involves partitioning the data into groups of similar
objects -- and, under various disguises, it is the fundamental problem
underlying most machine learning applications. The (theoretically)
well studied variants include $k$-median, $k$-means and $k$-center
clustering.  But many other variants of the clustering problem have
been subject of a long line of research \cite{dhs-pc-01}.

A clustering problem is often formalized as a constrained minimization
problem of a cost function. The cost function captures the similarity
of the objects in the same cluster. By minimizing the cost function we
obtain a clustering of the data such that objects in the same cluster
are more similar (in some sense) to each other than to those in other
clusters. Many of this type of formalizations of clustering are both
computationally hard, and sensitive to noise -- often because the
number of clusters is a hard constraint.

\paragraph*{Clustering as a quality of service maximization.}
An alternative formalization of the clustering problem is as a
\emph{maximization} problem where the goal is to maximize the quality
of ``service'' the data gets from the facilities chosen.  As a
concrete example, consider a set of $n$ clients, and the problem is
building $k$ facilities. The quality of service a client gets is some
monotonically decreasing non-negative function of its distance to the
closest facility. As a concrete example, for a mobile client, this
quantity might be the bandwidth available to the client. We refer to
this problem as the \emph{$k$-service} problem.

Such a formalization of clustering has several advantages.  The first
is diminishing returns (a.k.a.\ submodularity) -- that is, the marginal
value of a facility decreases as one adds more facilities. This
readily leads to an easy constant approximation algorithm. %
A second
significant advantage is the insensitivity to outliers -- a few points
being far away from the chosen facilities are going to change the
target function by an insignificant amount (naturally, these outliers
would get little to no service).

\paragraph*{Formal problem statement: $k$-service.}
Given a set $\PS \subseteq \Re^d$ of $n$ points, a monotonically
decreasing function $\sFunc:\Re^+ \rightarrow \Re^+$, the goal is to
choose a set $\CS$ of $k$ centers (not necessarily among the $n$ given
points), that maximize
$\sum_{\pa \in \PS} \sFunc\bigl( \dSY{\pa}{\CS} \bigr)$, where
$\dSY{\pa}{\CS} = \min_{\ca \in \CS} \dY{\pa}{\ca}$.

\paragraph*{Our result.}
We obtain an $n^{\eps^{-O(d)}}$ time algorithm for this problem that
achieves $(1-\eps)$-\si{approxi}\-\si{mation} for points in $\Re^d$.

\paragraph*{Related work.}

Maximum coverage problems, such as partial coverage by disks, were
studied in the past \cite{jlwzz-nltas-18}. These problems can be
interpreted as a $k$-service problem, where the function is $1$ within
distance $r$ from a facility, and zero otherwise.  In particular,
Chaplick \etal \cite{cdrs-asgcp-18} showed a \PTAS for covering a
maximum number of points, out of a given set of disks in the
plane. Our result implies a similar result in higher dimensions,
except that we consider the continuous case (i.e., our results yields
a \PTAS for covering the maximum number of points using $k$ unit disks).
Cohen-Addad \etal \cite{ckm-lsyas-19} showed that local search leads
to a \PTAS for $k$-median and $k$-means clustering in low dimensions
(and also in minor-free graphs). In~\cite{cohen2018fast} it was shown that the local search for $k$-means can be made faster achieving the runtime of $n \cdot k \cdot (\log n)^{(d/\eps)^{O(d)}}$. In~\cite{cohen2018near} near-linear time approximation schemes were obtained for several clustering problems improving on an earlier work (in particular,~\cite{friggstad2019local}). The authors achieve the runtime of $2^{(1/\eps)^{O(d^2)}}n(\log n)^{O(1)}$.
The techniques from the above works do not seem to be able to give near-linear time solution for the $k$-service problem, unfortunately. For instance, consider the special case of the $k$-service problem with $k=1$ and where the service function is $1$ within distance $r$ from a facility, and zero otherwise (the maximum coverage problem as above). Even for this very special case of the problem there is no algorithm known running in time $n^{o(d)}$ where $d$ is the dimension of the underlying space. The special case of $k=1$ is a significant obstacle towards obtaining near-linear time algorithms for the $k$-service problem.

Another related line of work is on the kernel density estimation
problem where a set $P$ of $n$ points is given and the goal is to
preprocess $P$ such that for an arbitrary query point $c$ one can
efficiently approximate
$\sum_{\pa \in \PS}\sFunc\bigl( \dSY{\pa}{c} \bigr)$. The goal is to
answer such queries much faster than in $O(nd)$ time, which is just
the linear scan. For various service functions $\sFunc$ and distance
functions $\mathsf{d}$ significantly faster algorithms are
known~\cite{greengard1991fast, charikar2017hashing,
   backurs2018efficient}. Despite the similarity, however, finding a
point (center) $c$ that (approximately) maximizes the sum
$\sum_{\pa \in \PS}\sFunc\bigl( \dSY{\pa}{c} \bigr)$ seems to be a
much harder problem~\cite{georgescu2003mean}.\footnote{In particular, to find such a point,~\cite{georgescu2003mean} use a heuristic that iteratively computes the gradient (mean shift vector) to obtain a sequence of points that converge to a local maxima (mode) of the density.}  Our work can be seen as
a generalization of the latter problem where our goal is to pick $k$
centers instead of one center.

In a another work, Friggstad \etal \cite{friggstad2019approximation} addressed the clustering problem in the setting with outliers.

\paragraph*{Balanced divisions.}
One of the building blocks we need is balanced divisions for Voronoi
diagrams.  This is present in the recent paper of Cohen-Addad \etal
\cite{ckm-lsyas-19}. The idea of balanced divisions seems to go back
to the work of Cohen-Addad and Mathieu \cite{cm-elsgo-15}.  Chaplick
\etal \cite{cdrs-asgcp-18} also prove a similar statement for planar
graphs.

For the sake of completeness, we include the proof of the desired
balanced divisions we need in \apndref{balanced}.  Both here and
\cite{ckm-lsyas-19} uses the Voronoi separator of Bhattiprolu and
Har-Peled \cite{bh-svdls-16} as the starting point to construct the
desired divisions. The Voronoi divisions we construct here are
slightly stronger than the one present in \cite{ckm-lsyas-19} -- all
the batches in the division are approximately of the same size, and
each one has a small separator from the rest of the point set.

\paragraph*{Clustering and submodularity.}
Work using submodularity in clustering includes Nagano \etal
\cite{nki-macc-10} and Wei \etal \cite{wiwbb-mrasp-15}.  Nagano \etal
\cite{nki-macc-10} considers the problem of computing the multi-way
cut, that minimizes the average cost (i.e., number of edges in the cut
divided by the number of clusters in the cut). Wei \etal
\cite{wiwbb-mrasp-15} also studies such partitions with average cost
target function. These works do not have any direct connection to what
is presented here, beyond the usage of submodularity.

\paragraph*{Paper organization.}
We define the problem formally in \secref{prelim}, and review some
necessary tools including submodularity and balanced subdivisions.
\secref{good:exchange} describes how to find a good exchange for the
current solution.  We describe the local search algorithm in
\secref{local:search}. The main challenge is to prove that if the
local search did not reach a good approximation, then there must be a
good exchange that improves the solution -- this is proved in
\secref{correctness}.

\section{Preliminaries}
\seclab{prelim}%

\paragraph*{Notations.}

In the following, we use $X+ x$ and $X-x$ as a shorthand for
$X \cup \brc{x}$ and $X \setminus \brc{x}$, respectively.

Given a point $\pa \in\Re^d$, and a set $\CSA \subseteq \Re^d$, we
denote by $\dSY{\pa}{\CSA} = \min_{\cb \in \CSA} \dY{\pa}{\cb}$ the
\emph{distance} of $\pa$ from $\CSA$. A point in $\CSA$ realizing this
distance is the \emph{nearest-neighbor} to $\pa$ in $\CSA$, and is
denoted by $\nnY{\pa}{\CSA} = \arg \min_{\ca \in \CSA} \dY{\ca}{\pa}$.

\subsection{Service function and problem statement}

A service function is a monotonically non-increasing function
$\sFunc: \Re^+ \rightarrow \Re^+$. In the following, %
given $x \geq 0$, assume that one can compute, in constant time, both
$\sFunc(x)$ and $\sFunc^{-1}(x)$. Given a point $\pa \in \Re^d$, and a
center $\ca \in \Re^d$, the quality of service that $\ca$ provides
$\pa$ is $ \profitY{\ca}{\pa} = \sFunc\bigl( \dY{\pa}{\ca}
\bigr)$. For a set of centers $\CS$, the quality of service it
provides to $\pa$ is
\begin{equation*}
    \profitY{\CS}{\pa}%
    =%
    \max_{\ca \in \CS} \profitY{\ca}{\pa}%
    =%
    \profitY{\nnY{\pa}{\CS}}{\pa}.
\end{equation*}
The service to $\PS$ provided by the set of centers $\CS$, or
just profit, is
\begin{equation*}
    \profitX{\CS}%
    =%
    \profitY{\CS}{\PS}%
    =%
    \sum_{\pa \in \PS} \profitY{\CS}{\pa }.
\end{equation*}

In the $k$-service problem, the task is to compute the set
$\Copt$ of $k$ points that realizes
\begin{equation*}
    \optY{k}{\PS}%
    =%
    \max_{\CS \subseteq \Re^d, \cardin{\CS} = k} 
    \profitY{\CS}{\PS}.%
\end{equation*}

\subsection{Submodularity}

The above is a submodular optimization problem. Indeed, consider a
center point $\ca$, and a set of centers $\CS$. The cell of
$\ca$, in the Voronoi partition induced by $\CS$, is
\begin{equation*}
    \clusterY{\ca}{\CS}%
    =%
    \Set{ \pa \in \PS}{\!\bigl. \dY{\ca}{\pa} < \dSY{\pa}{\CS - \ca}}.
\end{equation*}

\begin{definition}
    The marginal value of $\ca$ is
    \begin{equation*}
        \mrgY{\ca}{\CS}%
        =%
        \profitY{\CS + \ca }{\PS} - \profitY{\CS - \ca  }{\PS} 
        =%
        \sum_{\pa \in \clusterY{\ca}{\CS+\ca}}
        \pth{ \profitY{\CS+\ca}{\pa} - \profitY{\CS-\ca}{\pa} }.
    \end{equation*}
    In words, this is the increase in the service that one gets from
    adding the center $\ca$.
\end{definition}

For two sets of centers $\CSA \subseteq \CS$, and a center $\ca$,
observe that
$\clusterY{\ca}{\CS+\ca} \subseteq \clusterY{\ca}{\CSA+\ca}$. In
particular, we have
\begin{align*}
  \mrgY{\ca}{\CSA}%
  &=%
    \sum_{\pa \in \clusterY{\ca}{\CSA+\ca}}
    \pth{ \profitY{\ca}{\pa} - \profitY{\CSA }{\pa} }
    \geq%
    \sum_{\pa \in \clusterY{\ca}{\CS+\ca}}
    \pth{ \profitY{\ca}{\pa} - \profitY{\CSA}{\pa} }
  \\&%
  \geq%
  \sum_{\pa \in \clusterY{\ca}{\CS+\ca}}
  \pth{ \profitY{\ca}{\pa} - \profitY{\CS}{\pa} }
  =%
  \mrgY{\ca}{\CS}.%
\end{align*}
This property is known as submodularity.

\subsection{Balanced divisions}

For a point set $\PS \subseteq \Re^d$, the Voronoi diagram of
$\PS$, denoted by $\VorX{\PS}$ is the partition of space into convex
cells, where the Voronoi cell of $\pa \in \PS$ is
\begin{equation*}
    \VorCell{\pa}{\PS}%
    =%
    \Set{\pb \in \Re^d}{\dY{\pb}{\pa} \leq \distSetY{\pb}{\PS - \pa}},
\end{equation*}%
where $\distSetY{\pb}{\PS} = \min_{\pc \in \PS} \dY{\pb}{\pc}$ is the
distance of $\pb$ to the set $\PS$, see \cite{bcko-cgaa-08} for more
details on Voronoi diagrams.

\begin{definition}
    \deflab{def2}%
    Let $\PS$ be a set of points in $\Re^d$, and $\PS_1$ and $\PS_2$
    be two disjoint subsets of $\PS$. The sets $\PS_1$ and $\PS_2$ are
    Voronoi separated in $\PS$ if for all $\pa_1 \in \PS_1$ and
    $\pa_2 \in \PS_2$, we have that their Voronoi cells are disjoint
    -- that is,
    $\VorCell{\pa_1}{\PS} \cap \VorCell{\pa_2}{\PS} = \emptyset$. That
    is, the Voronoi cells of the pointsets are non-adjacent.
\end{definition}

\begin{definition}
    Given a set $\PS$ of $n$ points in $\Re^d$, a set of pairs
    $\brc{ (\batch_1, \BSet_1), \ldots, (\batch_m,\BSet_m)}$ is a
    Voronoi $\bSize$-division of $\PS$, if for all $i$, we
    have \smallskip%
    \begin{compactenumi}
        \item $\batch_1, \ldots, \batch_m$ are disjoint,
        \item $\bigcup_j \batch_j = \PS$,
        \item pointset $\BSet_i$ Voronoi separates $\batch_i$ from
        $\PS \setminus \batch_i$ in the Voronoi diagram of
        $\PS \cup \BSet_i$ in the sense of \defref{def2}, and
        \item $\cardin{\batch_i} \leq \bSize$.
    \end{compactenumi}
    \smallskip%
    The set $\batch_i$ is the $i$\th batch, and $\BSet_i$ is
    its boundary.
\end{definition}

A balanced coloring is a coloring
$\chi: \PS \rightarrow \{-1,+1\}$ of $\PS$ by $\pm 1$, such that
$\chi(\PS) = \sum_{\pa \in \PS} \chi(\pa) = 0$. For a set
$X \subseteq \PS$, its discrepancy is $\cardin{ \chi(X)}$.  We
need the following balanced $\bSize$-division result.  Since this
result is slightly stronger than what is available in the literature,
and is not stated explicitly in the form we need it, we provide a
proof for the sake of completeness in \apndref{balanced}.

\SaveContent{\ThmDivBalancedStatement}%
{%
   Given a set $\PS$ of $n$ points in $\Re^d$, parameters
   $\delta \in (0,1)$ and $\bSize = \Omega(1/\delta^{d+1})$, and a
   balanced coloring $\chi$ of $\PS$, one can compute in polynomial
   time, a Voronoi $\bSize$-division
   $\Div = \brc{ (\batch_1, \BSet_1), \ldots, (\batch_m,\BSet_m)}$,
   such that the following holds: \smallskip%
   \begin{compactenumA}
       \item $\bigcup \batch_i = \PS$, and the batches
       $\batch_1, \ldots, \batch_m$ are disjoint.
       
       \item $m = O( n/\bSize)$.
       
       \item For all $i$, we have the following properties:
       \begin{compactenumA}
           \item the set $\BSet_i$~Voronoi separates $\batch_i$ from
           $\PS \setminus \batch_i$.
           
           \item %
           $(1-\delta)\bSize \leq \cardin{\batch_i} \leq \bSize$
           (except for the last batch, which might be of size at least
           $(1-\delta)\bSize$, and at most size $2\bSize$).

           \item $\cardin{\BSet_i} \leq \delta \cardin{\batch_i}$.
           
           \item $|\chi(\batch_i)| \leq \delta \cardin{\batch_i}$.
       \end{compactenumA}
       
   \end{compactenumA}%
}%

\begin{theorem}
    \thmlab{div:balanced}%
    \ThmDivBalancedStatement{}
\end{theorem}

\section{Computing a good local exchange (if it exists)}
\seclab{good:exchange}%

\begin{lemma}[Computing a good single center]
    \lemlab{aprox:marginal:1}%
    Let $\PS$ be a set of $n$ points in $\Re^d$, and let $\CS$ be a
    set of $k$ centers. Given a parameter $\eps \in (0,1)$, one can
    $(1-\eps)$-approximate the center $\ca \in \Re^d$ that maximizes
    the marginal value in $(n/\eps)^{O(d)}$ time. Formally, we have
    $\mrgY{\ca}{\CS} \geq (1-\eps)\max_{\cb \in \Re^d \setminus
       \CS}\mrgY{\cb}{\CS}$.
\end{lemma}
\begin{proof}
    Let $\profitC = \profitY{\CS}{\PS}$,
    $\cc = \arg \max_{\cb \in \PS}\mrgY{\cb}{\CS}$,
    $\Delta = \mrgY{\cc}{\CS}$, and $\ub = \sFunc(0)$. Clearly, the
    profit of the optimal solution, after adding any number of centers
    to $\CS$ (but at least one), is somewhere in the interval
    $[\profitC + \Delta, n \ub] \subseteq [\profitC + \Delta, \profitC
    + n \Delta]$, which follows from $\ub \leq (\profitC/n)+\Delta$.
    For a point $\pa \in \PS$, let
    \begin{equation*}
        v(\pa,i)%
        =%
        \min \!\biga(\profitY{\CS}{\pa} +
        \ell(i), \, \ub \biga) \qquad \text{ where }
        \ell(i) = (1+\eps/4)^i \frac{\eps \Delta}{4n},
    \end{equation*}
    for $i=0, \ldots, N$, where $N=\ceil{16(\ln n)/\eps^2}$. Let
    $r(\pa, i) = \sFunc^{-1}\bigl(v(\pa,i)\bigr)$ (this is the radius
    from $\pa$ where a center provides service $v(\pa,i)$).

    Place a sphere of radius $r(\pa,i)$ around each point
    $\pa \in \PS$, for $i=0,\ldots, N$. Let $\Family$ be the resulting
    set of spheres. Compute the arrangement $\ArrX{\Family}$, and
    place a point inside each face of this arrangement. Let $\PSA$ be
    the resulting set of points.  Compute the point $\ca \in
    \PSA$ %
    realizing $\max_{\cb \in \PSA}
    \mrgY{\cb}{\CS}$, and return it as the desired new center.

    To show the correctness, consider the (open) face
    $F$ of $\ArrX{C}$, that contains $\copt$, where
    $\copt$ is the optimal center to be added. Let
    $\cb$ be any point of $\PSA$ in $F$. Let
    \begin{equation*}
        \mrgC(\cb, \pa)%
        =%
        \profitY{\cb \cup \CS}{\pa} -
        \profitY{\CS}{\pa}.
    \end{equation*}
    Define
    $\mrgC(\copt,\pa)$ similarly.  Clearly, we have $\mrgY{\cb}{\CS} =
    \sum_{\pa \in \PS} \mrgC(\cb,\pa)$.  Let
    $\PS_1$ be all the points $\pa$ of
    $\PS$ such that $\mrgC(\copt,\pa) \leq \mrgC(\cb,\pa) + \eps
    \Delta/(4n)$. Similarly, let $\PS_2$ be all the points
    $\pa$ of $\PS$, such that
    \begin{equation*}
        \mrgC(\copt,\pa) > \mrgC(\cb,\pa) + \eps \Delta/(4n).
    \end{equation*}
    For any point $\pa \in \PS_2$, by the choice of
    $N$, there exists an index $i$, such that $\ell(i) \leq
    \mrgC(\copt,\pa) < \ell(i+1)$. By the choice of
    $\cb$ from the arrangement, we have that $\mrgC(\cb,\pa) \geq
    \ell(i)$, which in turn implies that
    \begin{equation*}
        \mrgC(\cb,\pa) 
        \leq \mrgC(\copt,\pa)  < (1+\eps/4)
        \mrgC(\cb,\pa)
        \implies%
        \mrgC(\cb,\pa)  \geq (1-\eps/2)\mrgC(\copt,\pa).
    \end{equation*}
    We thus have the following
    \begin{align*}
      \mrgY{\cb}{\CS}
      &=%
        \sum_{\pa \in \PS} \mrgC(\cb,\pa) 
        = %
        \sum_{\pa \in \PS_1} \mrgC(\cb,\pa) 
        +\sum_{\pa \in \PS_2} \mrgC(\cb,\pa) \\
        & \geq %
        \sum_{\pa \in \PS_1} \bigl( \mrgC(\copt,\pa) - \frac{\eps\Delta}{4n}\bigr)
        +\sum_{\pa \in \PS_2}  (1-\eps/2)\mrgC(\copt,\pa)
      \\
      &\geq%
        (1-\eps/2)\sum_{\pa \in \PS} \mrgC(\copt,\pa)  - 
        \frac{\eps\Delta}{4}
        \geq %
        (1-\eps)\mrgY{\copt}{\CS}.
    \end{align*}
	
    The runtime follows from the observation that the number of faces
    in the arrangement is $(n/\eps)^{O(d)}$.~
\end{proof}

\begin{lemma}
    \lemlab{const:approx}%
    Given a set $\PS$ of $n$ points in the plane, and a parameter $k$,
    one can compute in polynomial time (i.e., $n^{O(d)}$) a constant
    approximation to $\popt = \optY{k}{\PS}$.
\end{lemma}
\begin{proof}
    Follows by using \lemref{aprox:marginal:1} in a greedy fashion $k$
    times, with $\eps=0.1$, to get a set of $k$ centers. The quality
    of approximation readily follows from known results about
    submodularity \cite{w-agass-82}. Indeed, let
    $v_i = \profitY{\CS_i}{\PS}$ be the service provided by the first
    $i$ centers computed.  By submodularity, and the quality guarantee
    of \lemref{aprox:marginal:1}, we have that
    $\mrgY{\ca_i}{\CS_{i-1}} \geq (1-\eps) (\popt - v_{i-1})/k$.  In
    particular, setting $\Delta_0= \popt$, and
    $\Delta_i = \popt - v_{i-1}$, we have that
    $\Delta_i \leq (1 - (1-\eps)/k) \Delta_{i-1}$.  As such,
    $\Delta_k \leq \exp\pth{ - k (1-\eps) /k }\Delta_0 =
    \popt/e^{\eps-1} \leq \popt/2$. Namely, we have
    $v_k \geq \popt/2$, as desired.
\end{proof}

For two sets of points $S$ and $\CS$ we define
$\mrgY{S}{\CS} = \profitY{\CS + S }{\PS} - \profitY{\CS - S }{\PS}$.

\begin{lemma}
    \lemlab{aprox:marginal:set}%
    Let $\PS$ be a set of $n$ points in $\Re^d$, and let $\CS$ be a
    set of $k$ centers. Given an integer $t\geq 1$, a parameter
    $\eps \in (0,1)$, in $(n/\eps)^{O(dt)}$ time, one can
    $(1-\eps)$-approximate the set $S \subset \Re^d$ with $|S|=t$ that
    maximizes the marginal value in $(n/\eps)^{O(td)}$ time. Formally,
    we have
    $\mrgY{S}{\CS} \geq (1-\eps)\max_{F \subset \Re^d,\
       |F|=t}\mrgY{F}{\CS}$.
\end{lemma}
\begin{proof}
    Consider the optimal set $F$ of size $t$. Denote it by
    $S^*$. Compute the same arrangement as in the proof of
    \lemref{aprox:marginal:1}. Let $F_1, \ldots, F_t$ be the faces of
    $\ArrX{C}$ that contain the $t$ points of $S^*$. Pick an arbitrary
    point from each $F_i$ and let $S$ be the resulting point set of
    size $t$. Define
    \begin{equation*}
        \mrgC(S, \pa)%
        =%
        \profitY{S \cup \CS}{\pa} -
        \profitY{\CS}{\pa}.
    \end{equation*}
    and similarly $\mrgC(S^*, \pa)$.
	
    As in the proof of \lemref{aprox:marginal:1}, let $P_1$ be all the
    points of $P$ such that
    $\mrgC(S^*,\pa) \leq \mrgC(S,\pa) + \eps \Delta/4n$ and $P_2$ be
    all the points of $P$ such that
    \begin{equation*}
        \mrgC(S^*,\pa) > \mrgC(S,\pa) + \eps \Delta/4n.
    \end{equation*}
    We also conclude that $\mrgC(S,\pa) \geq (1-\eps/2)\mrgC(S^*,\pa)$
    for all $p \in P_2$. We get
    \begin{align*}
      \mrgY{S}{\CS}
      &=%
        \sum_{\pa \in \PS} \mrgC(S,\pa) 
        = %
        \sum_{\pa \in \PS_1} \mrgC(S,\pa) 
        +\sum_{\pa \in \PS_2} \mrgC(S,\pa) 
      \\&
      \geq %
      \sum_{\pa \in \PS_1} \bigl( \mrgC(S^*,\pa) - \frac{\eps\Delta}{4n}\bigr)
      +\sum_{\pa \in \PS_2}  (1-\eps/2)\mrgC(S^*,\pa)
      \\&%
      \geq%
      (1-\eps/2)\sum_{\pa \in \PS} \mrgC(S^*,\pa)  - 
      \frac{\eps\Delta}{4}
      \geq %
      (1-\eps)\mrgY{S^*}{\CS}.
    \end{align*}
	
    The runtime follows from the observation that the number of faces
    in the arrangement is $(n/\eps)^{O(d)}$ and that it is sufficient
    to consider subsets of size $t$ of the faces.
\end{proof}

\section{The local search algorithm}
\seclab{local:search}

The algorithm starts with a constant approximation, using the
algorithm of \lemref{const:approx}. Next, the algorithm performs local
exchanges, as long as it can find a local exchange that is
sufficiently profitable.

Specifically, let $\bSize= O(1/\eps^{d})$, and let
$\xs = O(\bSize/\eps) = O(1/\eps^{d+1})$. Assume that one can
``quickly'' check given a set of $k$ centers $\CS$, whether there is a
local exchange of size $\xs$, such that the resulting set of centers
provides service $(1+\eps^2/(16k)) \profitC_{\mathrm{curr}}$, where
$\profitC_{\mathrm{curr}}$ is the service of the current solution. To
this end, the algorithm considers at most $k^\xs$ possible subsets of
the current set of centers that might be dropped, and for each such
subset, one can apply \lemref{aprox:marginal:set}, to compute
(approximately) the best possible centers to add. If all such subsets
do not provide an improvement, the algorithm stops.

\paragraph*{Running time analysis.}

The algorithm starts with a constant approximation. As such, there
could be at most $O( k/\eps^2)$ local exchanges before the algorithm
must terminate. Finding a single such exchange requires applying
\lemref{aprox:marginal:set} $k^\xs$ times.
\lemref{aprox:marginal:set} is invoked with $t
=\xs$. %
The resulting running time is
$k^{\xs}(n/\eps)^{O(d \xs)} = (n/\eps)^{O(d/\eps^{d+1})}$, where we
remember that $k\leq n$.

\section{Correctness of the local search algorithm}
\seclab{correctness}

Here we show that if the local algorithm has reached a local optimum,
then it reached a solution that is a good approximation to the optimal
solution.

\begin{remark}
    \remlab{shortcut}%
    In the following, we simplify the analysis at some points, by
    assuming that the local solution takes an exchange if it provides
    any improvement (the algorithm, however, takes an exchange only if
    it is a significant improvement).  Getting rid of the assumption
    and modifying the analysis is straightforward, but tedious.
\end{remark}

\subsection{Notations}
Let $\Lopt$ and $\Opt$ be the local and optimal set of $k$ centers.
Let $\US = \Lopt \cup \Opt$. Assign a point of $\US$ color $+1$ if it
is in $\Opt$ and $-1$ if it is in $\Lopt$ (for the sake of simplicity
of exposition assume no point belong to both sets).  Let $\constS$ be
a sufficiently large constant.  For $\delta = \eps/\constS$ and
$\bSize= O(1/\delta^{d+1})$, compute a $\bSize$-division
$\Div = \brc{ (\batch_1, \BSet_1), \ldots, (\batch_m,\BSet_m)}$ of
$\Lopt \cup \Opt$, using \thmref{div:balanced}.

Let $\Lopt_i = \batch_i \cap \Lopt$, $\LMX{i} = \Lopt_i \cup \BSet_i$,
$\Opt_i = \batch_i \cap \Opt$, $\OMX{i} = \Opt_i \cup \BSet_i$,
$\lSize_i = \cardin{\Lopt_i}$, and $\oSize_i = \cardin{\Opt_i}$, for
all $i$.  Let $\LM = \bigcup_i \LMX{i}$, and
$\OM = \bigcup_i \OMX{i}$.  Let $\BSet = \bigcup_i \BSet_i$. By
construction, we have that
$\sum_i \cardin{\BSet_i} \leq \delta 2 k \leq \eps k / 4$ if $\constS$
is a sufficiently large constant.

\subsection{Submodularity implies slow degradation}

The following is a well known implication of submodularity. We include
the proof for the sake of completeness.

\begin{lemma}
    \lemlab{slow:degradation}%
    Let $\CS$ be a set of $k$ centers. Then, for any $t \leq k$, there
    exists a subset $\CS' \subseteq \CS$ of size $t$, such that
    \begin{math}
        \profitX{\CS'} \geq \frac{t}{k} \profitX{\CS},
    \end{math}
    where $\profitX{\CS} = \profitY{\CS}{\PS}$.
\end{lemma}
\begin{proof}
    Let $\CS_0 = \CS$. In the $i$\th iteration, we greedily remove the
    point of $\CS_{i-1}$ that is minimizing the marginal
    value. Formally,
    \begin{equation*}
        \cb_i = \arg \min_{\ca \in \CS_{i-1}} \mrgY{\ca}{\CS_{i-1}-\ca}, 
    \end{equation*}
    and $\CS_i = \CS_{i-1} - \cb_i$. By submodularity, we have that
    $\mrgY{\cb_i}{\CS_{i-1}-\cb_i} \leq \profitX{\CS_{i-1}}
    /\cardin{\CS_{i-1}}$. As such, we have
    \begin{align*}
      \profitX{\CS_i}%
      &=%
        \profitX{\CS_{i-1}} -
        \mrgY{\cb_i}{\CS_{i-1}-\cb_i}
        \geq
        \pth{1 - \frac{1}{k-i+1}}  \profitX{\CS_{i-1}}
        =%
        \frac{k-i}{k-i+1}  \profitX{\CS_{i-1}}
      \\
      &\geq%
        \frac{k-i}{k-i+1} \cdot \frac{k-i+1}{k-i+1+1} \cdots
        \frac{k-1}{k} \profitX{\CS_0} 
        =%
        \frac{k-i}{k} \profitX{\CS}.
    \end{align*}
    The claim now readily follows by taking the set $\CS_{k-t}$.
\end{proof}

\subsection{Boundary vertices are not profitable}
First, we argue that adding the boundary points, does not increase the
profit/service significantly, for either the local or optimal
solutions.

\begin{lemma}
    \lemlab{b:not:profitable}%
    $\profitX{ \ts{\LM}\ts } \leq (1+\eps/4)\profitX{\Lopt}$ and %
    $\profitX{ \ts{\OM}\ts } \leq (1+\eps/4)\profitX{\Opt}$.
\end{lemma}
\begin{proof}
    Let $p = \cardin{\BSet}$.  Consider a point $\ca \in \BSet$, and
    observe that $\mrgY{\ca}{\Lopt} \leq \profitX{\Lopt}/k$. This is a
    standard consequence of submodularity and greediness/local
    optimality. To see that, order the centers of
    $\Lopt = \brc{\ca_1, \ldots, \ca_k}$ in an arbitrary order. Let
    $\mrgC_i = \mrgY{\ca_i}{\brc{\ca_1,\ldots, \ca_{i-1}}} \geq 0$,
    for $i=1,\ldots, k$, and observe that
    $\profitX{\Lopt} = \sum_{i=1}^k \mrgC_i$. As such, there exists an
    index $i$, such that $\mrgC_i \leq \profitX{\Lopt}/k$.  By
    submodularity, we have that
    \begin{math}
        \mrgY{\ca_i}{\Lopt - \ca_i}%
        \leq%
        \mrgY{\ca_i}{\brc{\ca_1,\ldots, \ca_{i-1}}} =%
        \mrgC_i%
        \leq%
        \profitX{\Lopt}/k,
    \end{math}
    and
    \begin{math}
        \mrgY{\ca}{\Lopt - \ca_i} \geq \mrgY{\ca}{\Lopt }.
    \end{math}
    Assume, for the sake of contradiction, that
    $\mrgY{\ca}{\Lopt} > \profitX{\Lopt}/k$.  We have that
    \begin{align*}
      \profitX{ \Lopt - \ca_i + \ca} %
      &=%
        \profitX{\Lopt} - \mrgY{\ca_i}{\Lopt - \ca_i} +
        \mrgY{\ca}{\Lopt - \ca_i}
        \geq%
        \profitX{\Lopt} -\frac{\profitX{\Lopt}}{k}  + \mrgY{\ca}{\Lopt}
      \\&
      > %
      \profitX{\Lopt} -\frac{\profitX{\Lopt}}{k}  +
      \frac{\profitX{\Lopt}}{k}%
      =%
      \profitX{\Lopt}.
    \end{align*}
    But the local search algorithm considered this swap, which means
    that $\Lopt - \ca_i + \ca$ can not be more profitable than the
    local solution. A contradiction (see \remref{shortcut}).

    Setting $\BSet = \{\cb_1, \ldots, \cb_p\}$, we have
    \begin{equation*}
        \profitX{ \ts{\LM}\ts }%
        =%
        \profitX{ \Lopt } + \sum_{i=1}^p\mrgY{\cb_i}{ \Lopt +
           \cb_1+ \cdots + \cb_{i-1}}
        \leq%
        \profitX{ \Lopt } + \sum_{i=1}^p \mrgY{\cb_i}{ \Lopt}
        \leq%
        \profitX{ \Lopt } + p \profitX{\Lopt}/k
        \leq%
        (1+\eps/4) \profitX{ \Lopt },
    \end{equation*}
    since $p = \cardin{\BSet} \leq \eps k /4$. %

    The second claim follows by a similar argument.
\end{proof}

\subsection{If there is a gap, then there is a swap}

The contribution of the clusters $\Lopt_i$ and $\Opt_i$ is
\begin{equation}
    \mrgC \Lopt_i = 
    \mrgY{\Lopt_i}{\LM \setminus \Lopt_i}.
    \qquad\text{and}\qquad%
    \mrgC \Opt_i = 
    \mrgY{\Opt_i}{\OM \setminus \Opt_i},
    \eqlab{contrib}%
\end{equation}
respectively.  Notice that, because of the separation property, the
points in $\PS$ that their coverage change when we move from
$\LM \setminus \Lopt_i$ to $\LM$, are points that are served by
$\BSet_i \subseteq \LM \setminus \Lopt_i$ (same holds for
$\OM \setminus \Opt_i$ and $\OM$).

This implies that
\begin{equation*}
    \mrgY{\Lopt}{\BSet} = \sum_i \mrgC \Lopt_i
    \qquad\text{and}\qquad%
    \mrgY{\Opt}{\BSet} = \sum_i \mrgC \Opt_i.
\end{equation*}

In the following, we assume that
\begin{math}
    \profitX{\Lopt} < (1-\eps) \profitX{\Opt}.
\end{math}
By \lemref{b:not:profitable} this implies that
\begin{math}
    \profitX{ \ts{\LM}\ts } \leq (1+\eps/4)\profitX{\Lopt} <
    (1+\eps/4)(1-\eps) \profitX{\Opt} \leq%
    (1-\eps/2) \profitX{\Opt} \leq (1-\eps/2) \profitX{\OM}.
\end{math}
As such, we have
\begin{align*}
  \profitX{\LM} < (1-\eps/2) \profitX{\OM}
  \implies
  &
    \profitX{\BSet} + \mrgY{\Lopt}{\BSet}
    <%
    (1-\eps/2) \pth{ \profitX{\BSet} + \mrgY{\Opt}{\BSet}}
  \\%
  \implies
  &%
    \mrgY{\Lopt}{\BSet}
    <%
    (1-\eps/2)  \mrgY{\Opt}{\BSet}
    - (\eps/2)\profitX{\BSet} 
  \\%
  \implies
  &%
    \mrgY{\Lopt}{\BSet}
    <%
    (1-\eps/2)  \mrgY{\Opt}{\BSet}
  \\%
  \implies
  &%
    (\eps/2) \mrgY{\Opt}{\BSet} 
    <%
    \sum_i (\mrgC \Opt_i - \mrgC \Lopt_i).
\end{align*}
By averaging, this implies that there exists an index $t$, such that
\begin{align}
  & \mrgC \Opt_t - \mrgC \Lopt_t 
    > %
    \frac{\eps}{2k}\mrgY{\Opt}{\BSet} 
  >%
  \frac{\eps}{2k}\bigl(\mrgY{\Opt}{\BSet}-\mrgY{\Lopt}{\BSet}\bigr)%
  \\
  &=%
  \frac{\eps}{2k}\bigl(\profitX{\OM}-\profitX{\LM}\bigr)%
  \geq%
  \frac{\eps}{2k} \cdot \frac{\eps}{2} \profitX{\OM}
  \geq%
  \frac{\eps^2}{4k} \profitX{\OM},
  \eqlab{good:exchange}%
\end{align}
where in the second to last inequality we use that
$\profitX{\LM} < (1-\eps/2) \profitX{\OM}$. Namely, there is a batch
where the local and optimal solution differ significantly.

\subsubsection{An unlikely scenario}
Assume that
$\cardin{\Lopt_t} \geq \cardin{\Opt_t} + \cardin{\BSet_t}$. We then
have that
\begin{equation*}
    \profitX{\Lopt +\BSet_t - \Lopt_t + \Opt_t }%
    =%
    \profitX{\Lopt +\BSet_t} - \mrgC{\Lopt_t} + \mrgC \Opt_t
    \geq%
    \profitX{\Lopt}  + \frac{\eps^2}{4k } \profitX{\OM}.
\end{equation*}
But this is impossible, since the local search algorithm would have
performed the exchange $\Lopt +\BSet_t - \Lopt_t + \Opt_t $, since
$\cardin{\BSet_t} + \cardin{\Lopt_t} + \cardin{\Opt_t}$ is smaller
than the size of exchanges considered by the algorithm.

\subsubsection{The general scenario}

\begin{lemma}
    \lemlab{shrink}%
    There exists a subset $Y \subseteq \Opt_t$, such that
    $\cardin{\Lopt_t} \geq \cardin{Y} + \cardin{\BSet_t}$, and
    \begin{equation*}
        \mrgY{Y}{\OM \setminus \Opt_t}
        \geq%
        \mrgC \Lopt_t + \frac{\eps^2}{8k}\profitX{\OM},
    \end{equation*}
    see \Eqref{contrib}.
\end{lemma}
\begin{proof}
    We have that
    $(\eps/2) \mrgY{\Opt}{\BSet} < \sum_i (\mrgC \Opt_i - \mrgC
    \Lopt_i)$. Subtracting $(\eps/8) \mrgY{\Opt}{\BSet}$ from both
    sides implies that
    $(\eps/4) \mrgY{\Opt}{\BSet} < \sum_i ((1-\eps/8)\mrgC \Opt_i -
    \mrgC \Lopt_i)$. This in turn implies that there exists $t$ such
    that
    $(1-\eps/8)\mrgC \Opt_t - \mrgC \Lopt_t > (\eps/(4k))
    \mrgY{\Opt}{\BSet}$. Arguing, as above, we have that
    $(1-\eps/8)\mrgC \Opt_t - \mrgC \Lopt_t >
    (\eps^2/(8k))\profitX{\OM}$.

    Consider the following (submodular) function
    \begin{equation*}
        f(X) =
        \mrgY{X}{\OM \setminus \Opt_t},
    \end{equation*}
    By \lemref{slow:degradation} (or more precisely arguing as in this
    lemma), we have that there exists a set $Y \subset \Opt_t$, such
    that $\cardin{Y} = (1-\eps/8) \cardin{\Opt_t}$ and
    $f(Y) \geq (1-\eps/8) f(\Opt_t) = (1-\eps/8) \mrgC \Opt_t$.  As
    such, we have that
    \begin{equation*}
        \mrgY{Y}{\OM \setminus \Opt_t}
        \geq%
        (1-\eps/8)  \mrgC \Opt_t 
        \geq%
        \mrgC \Lopt_t + (\eps^2/8k)\profitX{\OM}.
    \end{equation*}

    As for the size of $Y$. Observe that by \thmref{div:balanced}, we
    have
    \begin{math}
        \cardin{\BSet_i}\leq \frac{\eps}{\constS}(\cardin{\Opt_i} +
        \cardin{\Lopt_i})
    \end{math}
    and
    \begin{math}
        \cardin{\bigl. \cardin{\Opt_i} - \cardin{\Lopt_i}} \leq
        \frac{\eps}{\constS}(\cardin{\Opt_i} + \cardin{\Lopt_i}).
    \end{math}
    This readily implies that
    \begin{math}
        \cardin{\Opt_i} \leq%
        ( 1+4 \eps/\constS)\cardin{\Lopt_i},
    \end{math}
    and
    \begin{math}
        \cardin{\Lopt_i} \leq%
        ( 1+4 \eps/\constS)\cardin{\Opt_i},
    \end{math}
    if $\constS$ is sufficiently large.  As such, we have that
    \begin{equation*}
        \cardin{\Lopt_t}%
        \geq%
        \frac{\cardin{\Opt_t}}{1+4 \eps/\constS}
        \geq 
        \pth{1-4 \eps/\constS} \cardin{\Opt_t}
        =%
        \pth{1-\eps/8}\cardin{\Opt_t} + (\eps/8 - 4 \eps/\constS) \cardin{\Opt_t}
        \geq%
        \cardin{Y} + \frac{\eps}{16} \cardin{\Opt_t} 
        \geq%
        \cardin{Y} + \cardin{\BSet_i},
    \end{equation*}
    if $\constS \geq 64$.
\end{proof}

\begin{lemma}
    The local search algorithm computes a $(1-\eps)$-approximation to
    the optimal solution.
\end{lemma}
\begin{proof}
    If not, then, arguing as above, there must be a batch for which
    there is an exchange with profit at least
    $(\eps^2/4k) \profitX{\OM}$ (see \Eqref{good:exchange}). By
    \lemref{shrink}, we can shrink the optimal batch $\Opt_t$, such
    that the exchange becomes feasible, and is still profitable (the
    profit becomes $(\eps^2/8k) \profitX{\OM}$). But that is
    impossible, since by arguing as above (i.e., the unlikely
    scenario), we have that this swap would result in a better local
    solution, and the exchange is sufficiently small to have been
    considered.  Specifically, the local search algorithm uses
    \lemref{aprox:marginal:set}, say with $\eps=1/2$, ensures that the
    local search algorithm would find an exchange with half this
    value, and would take it.  A contradiction.
\end{proof}
 
\subsection{The result}

\begin{theorem}
    Let $\PS$ be a set of $n$ points in $\Re^d$, let $\eps \in (0,1)$
    be a parameter, let $\sFunc: \Re^+ \rightarrow \Re^+$ be a service
    function, and let $k\leq n$ be an integer parameter. One can compute, in
    $(n/\eps)^{O(d/\eps^{d+1})}$ time, a set of $k$ centers $\CS$,
    such that $\profitY{\CS}{\PS} \geq (1-\eps) \optY{k}{\PS}$, where
    $\optY{k}{\PS}$ denotes the optimal solution using $k$ centers.
\end{theorem}

\section{Discussion}

We presented an algorithm that runs in polynomial time for any
constant $\eps>0$ and any constant dimension $d$ and achieves a
$(1-\eps)$-approximation. %
The dependency on the dimension $d$ is doubly exponential, however. A
natural question is whether the dependency on the dimension $d$ in the
runtime can be improved. Perhaps by considering some special cases of
the problem, for more specific service function, such as
$\sFunc( \ell ) = \frac{1}{1+\ell}$ or
$\sFunc( \ell ) = \frac{1}{1+\ell^2}$ (i.e., the service quality drops
roughly linearly or quadraticly with the distance).

Our algorithm finds a subset $\CS \subseteq \Re^d$ of size $k$ that
approximately maximizes the objective function. A close variant of the
problem asks to find a subset $\CS \subseteq \Re^d$ of size $k$ that
has an additional constraint that $\CS \subseteq P$. Can we obtain an
algorithm for this problem with the same asymptotic runtime for $d>2$?
The main difficulty is that we would need a variant of
\thmref{div:balanced} with an additional property
$\BSet_i \subseteq P$ for all $i$ but such a division does not exist
even for $d=3$. (For $d=2$ using planar graph divisions on the Voronoi
diagram directly implies the desired result.)

Finally, one can ask a similar question about clustering in
graphs. Specifically, given a graph on $n$ vertices $P$, we would like
to select $k$ vertices $C$ that approximately maximizes
$\sum_{\pa \in \PS}\min_{c \in C}\sFunc\bigl( \dSY{\pa}{c} \bigr)$,
where $\dSY{\pa}{c}$ is the shortest-path distance from $p$ to
$c$. Can we achieve a polynomial time algorithm for arbitrary small
constant $\eps>0$ if the graph is planar or come from some other class
of graphs?

%

%
%

%

%

\newcommand{\etalchar}[1]{$^{#1}$}
 \providecommand{\CNFX}[1]{ {\em{\textrm{(#1)}}}}


%
%
%

\appendix

\section{Balanced Voronoi division}
\apndlab{balanced}

We need the following variant of a result of Bhattiprolu and Har-Peled
\cite{bh-svdls-16}.

\begin{theorem}
    \thmlab{separator:main}%
    Let $\PS$ be a set of $n'$ points in $\Re^d$, where every point
    has a positive integer weight, such that the total weight of the
    points is $n$, and let $\bSize$ be parameter. Furthermore, assume
    that no point has weight that exceeds $\bSize$. Then, one can
    compute, in expected $O(n)$ time, a ball $\ball$, and a set
    $\SepSet$ that lies on the boundary of $\ball$, such that
    \smallskip%
    \begin{compactenumi}
        \item $\cardin{\SepSet} \leq \constA \bSize^{1-1/d}$,
        \item the total weight of the points of $\PS$ inside $\ball$
        is at least $\bSize$ and at most $ \constB \bSize$,
                
        \item $\SepSet$ is a Voronoi separator of the points of $\PS$
        inside $\ball$ from the points of $\PS$ outside $\ball$.
    \end{compactenumi}
    \smallskip%
    Here $\constA, \constB > 0$ are constants that depends only on the
    dimension $d$.
\end{theorem}

The points of the separator $\SepSet$ are guards.

\subsection{A division using the above separator}

\subsubsection{Algorithm}

We start with a set $\PS$ of $n$ points, and a parameter $\bSize$.
The idea is to repeatedly extract a set of weight (roughly) $\bSize$
from the point set, separate it, remove it, and put the set of guards
associated with it back into the set.

To this end, let $\bSize$ be a parameter, such that
\begin{equation*}
    \constA \bSize^{1-1/d} < \bSize/8%
    \iff%
    8\constA  < \bSize^{1/d}%
    \iff%
    \bSize > (8\constA)^d,
\end{equation*}
where $\constA$ is the constant from \thmref{separator:main}.

For an unweighted set of points $X$ and a real number $\tau >0$, let
$\tau*X$ denote the set of points, where every points has weight
$\tau$.

The algorithm for constructing the division is the following:

\begin{compactenumN}
    \item $\PS_0 \leftarrow \PS$. Initially all the points in $\PS_0$
    have weight $1$.
    \item $i \leftarrow 1$.

    \item While $\PS_{i-1}$ has total weight larger than $\bSize$ do:
    \begin{compactenumN}
        \item $(\ball_i,\SepSet_i) \leftarrow $ ball and separator
        computed by \thmref{separator:main} for $\PS_{i-1}$ with
        parameter $\bSize$.
        
        \item $\inSet_i \leftarrow \PS_{i-1} \cap \ball_i$. %
        \hfill\CC{// All points inside ball to be removed}

        \item $\GS_i \leftarrow \inSet_i \setminus \PS$. \hfill \CC{//
           The old guards in the ball}

        \item $\batch_i = \PS \cap \inSet_i$ \hfill\CC{// The batch of
           original points}

        \item
        $\PS_i = \pth{\PS_{i-1} \setminus \ball_i} \cup (\tau_i *
        \SepSet_i)$, where
        $\tau_i = \ceil{(\bSize/4)/ \cardin{\SepSet_i}}$.

        \item $\BSet_i = \SepSet_i \cup \GS_i $ \hfill \CC{// The set
           of guards for the batch $\batch_i$}
        
        \item $i \leftarrow i+1$.
    \end{compactenumN}

    \item $m \leftarrow i$
    \item $\batch_m = \PS_{m-1} \cap \PS$, and
    $\BSet_m = \PS_{m-1} \setminus \batch_m$.

    \item Return
    $\Div = \brc{ (\batch_1, \BSet_1), \ldots, (\batch_m,\BSet_m)}$.
\end{compactenumN}

\subsubsection{Analysis}

\begin{lemma}
    Consider the Voronoi diagram of $\VorX{\PS \cup \BSet_i}$. There
    is no common boundary in this Voronoi diagram between a cell of a
    point of $\batch_i$ and a cell of a point of
    $\PS \setminus \batch_i$.
\end{lemma}
\begin{proof}
    Consider a point $\pa$ that is in equal distance to a point
    $\cb \in \batch_i$, and a point $\cc \in \PS \setminus \batch_i$,
    and furthermore, all other points of
    $\PS \setminus \brc{ \cb, \cc}$ are strictly further away from
    $\pa$.

    The claim is that
    $\distSetY{\pa}{\BSet_i} < \dY{\pa}{\cb} = \dY{\pa}{\cc}$. Namely,
    the region of common boundary between $\cb$ and $\cc$ in
    $\VorX{\PS}$ is completely covered by cells of $\BSet_i$ in
    $\VorX{\PS \cup \BSet_i}$.

    If $\cc \in \PS_i$, then $\SepSet_i$ separates (in the Voronoi
    interpretation) $\cb \in \batch_i \subseteq \inSet_i$ from all the
    points of $\PS_i \cap \PS \ni \cc$, which implies the claim.

    As such, it must be that $\cc \in \batch_j$, for some $j <
    i$. Namely, there is a guard $\cc_j \in \SepSet_j$, that separates
    $\cc$ from $\cb$, and its cell contains $\pa$. That is
    $\dY{\pa}{\cc_j} < \dY{\pa}{\cb}$, and $\cc_j \in \PS_j$. If
    $\cc_j \in \BSet_i$ then the claim holds.

    Otherwise, we apply the same argument again, this time to $\cc_j$
    and $\cb$. Indeed, $\cc_j$ was removed (from $\PS_k$) in some
    iteration $k$, such that $j < k <i$. Namely $\cc_j \in \ball_k$,
    and $\cb \notin \ball_k$. The point $\pa$ is closer to $\cc_j$
    then to $\cb$. If $\pa \in \ball_k$ then there is a guard
    $\cc_k \in\SepSet_k$ that is closer to $\pa$ than $\cb$, by the
    separation property.  Otherwise, it is easy to verify that the
    Voronoi cells of the guards of $\SepSet_k$ in
    $\VorX{\PS_{k-1} \cup \SepSet_k}$ cover completely the portion of
    the Voronoi cells of points in $\PS_{k-1}\cap \ball_k$ outside
    $\ball_k$, in the Voronoi diagram $\VorX{\PS_{k-1}}$. This readily
    implies that there is a closer guard $\cc_k \in \SepSet_k$ to
    $\pa$ than $\cc_j$.  In either case, we continue the argument
    inductively on $(\cc_k,\cb)$.

    By finiteness, it follows that there must be a guard
    $\cc' \in \BSet_i$ that is closer to $\pa$ than $\cb$, which
    implies the claim.
\end{proof}

\begin{observation}
    (A) For all $i$, we have
    \begin{math}
        \tau_i%
        =%
        \ceil{\bigl.\frac{\bSize/4}{ \cardin{\SepSet_i}}}%
        \geq%
        \frac{\bSize/4}{ \cardin{\SepSet_i}}%
        \geq%
        \frac{\bSize/4}{\constA \bSize^{1-1/d}}%
        \geq%
        \frac{\bSize^{1/d}}{4\constA }.
    \end{math}

    (B) As such, for all $i$, $\inSet_i$ contains at most
    $\constB \bSize / \min_j \alpha_j = O( \bSize^{1-1/d} )$ points
    that are not in $\PS$. That is, we have
    $\cardin{\inSet_i \setminus \batch_i} =O(\bSize^{1-1/d})$.

    (C) It follows that
    $\cardin{\BSet_i} = \cardin{\inSet_i \setminus \batch_i} +
    \cardin{\SepSet_i} = O(\bSize^{1-1/d})$.
\end{observation}

\begin{lemma}
    \lemlab{div:1}%
    Given a set $\PS$ of $n$ points in $\Re^d$, and a parameter
    $\bSize$, one can compute in polynomial time, a division
    $\Div = \brc{ (\batch_1, \BSet_1), \ldots, (\batch_m,\BSet_m)}$,
    such that the following holds: \smallskip%
    \begin{compactenumA}
        \item $\bigcup \batch_i = \PS$, and the clusters
        $\batch_1, \ldots, \batch_m$ are disjoint.
        
        \item $m = O(n/\bSize)$.
        
        \item For all $i$, we have the following properties:
        \begin{compactenumA}
            \item the set $\BSet_i$ separates $\batch_i$ from
            $\PS \setminus \batch_i$.
            
            \item \itemlab{C:ii} $\cardin{\batch_i} = O(\bSize)$.
            
            \item $\cardin{\BSet_i} = O(\bSize^{1-1/d})$.
        \end{compactenumA}

        \smallskip%
        \item For $\BSet = \bigcup_i \BSet_i$, we have that
        $\cardin{\BSet} = O( n/\bSize^{1/d})$.
    \end{compactenumA}%
    \smallskip%
    Furthermore, one can modify the above construction, so that
    \itemref{C:ii} is replaced by
    $\cardin{\batch_i} = \Theta(\bSize)$.
\end{lemma}
\begin{proof}
    For the bound on number of clusters, observe that every iteration
    of the algorithm reduces the weight of the working set $\PS_i$ by
    at least $\bSize/2$ -- indeed, the weight of $\batch_i$ is at
    least $\bSize$, and the total weight of points of $\SepSet_i$
    (after multiplying their weight by $\tau_i$) is at most
    $\bSize/2$. Thus implying the claim.

    All the other claims are either proved above, or readily follows
    from the algorithm description.

    The modification of \itemref{C:ii} follows by observing that we
    can merge clusters, and there are $\Theta(n/\bSize)$ clusters with
    $\Omega(\bSize)$ points of $\PS$, by averaging. As such, one can
    merge $O(1)$ clusters that have $o(\bSize)$ points of $\PS$ into a
    cluster that has $\Omega(\bSize)$ points of $\PS$, thus implying
    the modified claim.
\end{proof}

\bigskip%
\RestatementOf{\thmref{div:balanced}}{\ThmDivBalancedStatement{} }

\begin{proof}
    We compute a division of $\PS$ using \lemref{div:1}, with
    parameter $\bSize'= O(\delta \bSize) = \Omega(1/\delta^d)$, such
    that (i) the maximum size of a batch is strictly smaller than
    $\delta\bSize/2$, and (ii)
    $\constD (\bSize')^{1-1/d} < \delta \bSize'$, where $\constD$ is
    some prespecified constant. Let
    $\Div' = \brc{ (\batch_1', \BSet_1'), \ldots,
       (\batch_\tau',\BSet_\tau')}$, be the resulting division. Let
    $\balance_i = \chi(\BSet_i')$, and observe that
    $\sum_i \balance_i = \chi(\PS) = 0$.  There is a permutation $\pi$
    of the batches, such that for any prefix $j$, we have
    $|\sum_{i=1}^j \balance_{\pi(i)}| \leq \max_i \cardin{\batch_i}
    \leq \delta\bSize/2 = D$.  This follows readily by reordering the
    summation, such that one adds batches with positive (resp.,
    negative) balance if the current prefix sum is negative (resp.,
    positive), and repeating this till all the terms are
    used\footnote{This is the same idea that is used in the Riemann's
       rearrangement theorem.}.

    We break the permutation $\pi$ into minimum number of consecutive
    intervals, such that total size of batches in each interval is at
    least $(1-\delta)\bSize$. Merging the last two interval if needed
    to comply with the desired property. The batch formed by a union
    of an interval can have discrepancy at most
    $2D = 2(\delta\bSize/2) = \delta \bSize$, as desired.

    All the other properties follows readily by observing that merging
    batches, results in valid batches, as far as separation.
\end{proof}

\remove{%
   \subsection{Planar graphs}

\begin{lemma}
    \lemlab{graph_div:1}%
    Given a planar graph of $n$ vertices, and a parameter $\bSize$,
    one can compute in $O(n\log n)$ time, a division
    $\Div = \brc{ (\batch_1, \BSet_1), \ldots, (\batch_m,\BSet_m)}$,
    such that the following holds: \smallskip%
    \begin{compactenumA}
        \item $\bigcup \batch_i = \PS$, and the clusters
        $\batch_1, \ldots, \batch_m$ are disjoint.
        $\BSet_i \subseteq P$ for $i=1, \ldots, m$.
        
        \item $m = O(n/\bSize)$.
        
        \item For all $i$, we have the following properties:
        \begin{compactenumA}
            \item the set $\BSet_i$ separates $\batch_i$ from
            $\PS \setminus \batch_i$.
            
            \item \itemlab{C:ii:u} $\cardin{\batch_i} = O(\bSize)$.
            
            \item $\cardin{\BSet_i} = O(\sqrt{\bSize})$.
        \end{compactenumA}

        \smallskip%
        \item For $\BSet = \bigcup_i \BSet_i$, we have that
        $\cardin{\BSet} = O( n/\sqrt{\bSize})$.
    \end{compactenumA}%
    \smallskip%
    Furthermore, one can modify the above construction, so that
    \itemref{C:ii} is replaced by
    $\cardin{\batch_i} = \Theta(\bSize)$.
\end{lemma}
\begin{proof}
    Follows from Lemma 2 in \cite{federickson1987fast}.
\end{proof}

Using \lemref{graph_div:1} we can prove the following analog of
\thmref{div:balanced} for planar graphs.

\begin{lemma}
    \lemlab{graph_div:balanced}%
    Given a planar graph of $n$ vertices, and parameters $\delta$ and
    $\bSize = \Omega(1/\delta^{3})$, and a balanced coloring $\chi$ of
    $\PS$, one can compute in $O(n \log n)$ time, a division
    $\Div = \brc{ (\batch_1, \BSet_1), \ldots, (\batch_m,\BSet_m)}$,
    such that the following holds: \smallskip%
    \begin{compactenumA}
        \item $\bigcup \batch_i = \PS$, and the clusters
        $\batch_1, \ldots, \batch_m$ are disjoint.
        $\BSet_i \subseteq P$ for $i=1, \ldots, m$.
        
        \item $m = O( n/\bSize)$.
        
        \item For all $i$, we have the following properties:
        \begin{compactenumA}
            \item the set $\BSet_i$ separates $\batch_i$ from
            $\PS \setminus \batch_i$.
            
            \item \itemlab{C:ii:2:x}
            $(1-\delta)\bSize \leq \cardin{\batch_i} \leq \bSize$
            (except for the last batch, which might be of size at
            least $(1-\delta)\bSize$, and at most size $2\bSize$).

            \item $\cardin{\BSet_i} \leq \delta \cardin{\batch_i}$.

            \item $|\chi(\batch_i)| \leq \delta \cardin{\batch_i}$.
        \end{compactenumA}

    \end{compactenumA}%
\end{lemma}

}

\end{document}